\documentclass[pdflatex,sn-mathphys-num]{sn-jnl}

\usepackage{graphicx}%
\usepackage{multirow}%
\usepackage{amsmath,amssymb,amsfonts}%
\usepackage{amsthm}%
\usepackage{mathrsfs}%
\usepackage[title]{appendix}%
\usepackage{xcolor}%
\usepackage{textcomp}%
\usepackage{manyfoot}%
\usepackage{booktabs}%
\usepackage{algorithm}%
\usepackage{algorithmicx}%
\usepackage{algpseudocode}%
\usepackage{listings}%

\theoremstyle{thmstyleone}%
\newtheorem{theorem}{Theorem}[section]
\newtheorem{lemma}[theorem]{Lemma}%
\newtheorem{corollary}[theorem]{Corollary}%

\theoremstyle{thmstyletwo}%
\newtheorem{remark}[theorem]{Remark}%

\theoremstyle{thmstylethree}%

\numberwithin{equation}{section}

\def\N{{\mathbb N}}
\def\Z{{\mathbb Z}}
\def\R{{\mathbb R}}
\def\supp{\mathop{\rm supp} \nolimits}

\begin{document}

\title[Article Title]{Magnetic Pseudo-differential Operators with Hörmander Symbols Dominated by Tempered Weights}


\author*[1]{\fnm{Mikkel Hviid} \sur{Thorn}}\email{mikkelht@math.aau.dk}

\affil[1]{\orgdiv{Department of Mathematical Sciences}, \orgname{Aalborg University}, \orgaddress{\street{Thomas Manns Vej 23}, \city{Aalborg}, \postcode{9220}, \country{Danmark}}}


\abstract{We extend the matrix representation of magnetic pseudo-differential operators in a tight Gabor frame from \cite{CorneanHelfferPurice2018,CorneanHelfferPurice2024} to asymmetrical quantizations and smooth symbols dominated by a tempered weight (and not just decay/growth properties in the momentum variables). This leads to new results regarding the symbol calculus of such operators and their Schatten-class properties.}




\maketitle

\section{Introduction}

The gauge-invariant magnetic pseudo-differential calculus was first proposed in \cite{Mantoiu2004} and has since been developed into a rich theory \cite{Iftimie2007,Iftimie2009,Iftimie2010,Iftimie2010a,Athmouni2018}. In the papers \cite{CorneanHelfferPurice2018,CorneanHelfferPurice2024}, the magnetic Weyl quantization was studied for classical Hörmander symbols with uniform decay properties in the momentum variables for all derivatives by expanding the magnetic pseudo-differential operators in a tight Gabor frame. The resulting matrix representation led to simple and clear proofs of major results in this magnetic pseudo-differential calculus, see e.g. the magnetic Calderón-Vaillancourt theorem \cite[Theorem 3.7]{CorneanHelfferPurice2024} (originally \cite[Theorem 3.1]{Iftimie2007}) and Beals' criterion \cite[Theorem 3.8]{CorneanHelfferPurice2024} (originally \cite[Theorem 1.1]{Iftimie2010}). This work was inspired by \cite{Feichtinger1997,Grochenig2006}. Alternative matrix representations have also been considered in recent years \cite{CorneanGarde2019,Bachmann2025}, each with their own strengths and weaknesses. These matrix representations can help solve problems such as the question of spectral regularity when the magnetic field varies in strength \cite{Cornean2010,CorneanPurice2015} as in \cite{CorneanGarde2019}.

In this paper we study classical Hörmander symbols with uniform decay/growth properties for all derivatives where the decay/growth is dominated by a tempered weight (temperate weight \cite{Nicola2010}, order function \cite{Zworski2012}). For these symbol classes we derive a characterization of the matrices for the corresponding magnetic pseudo-differential operators (Theorem \ref{thm:matrix_rep}). These matrices have a well-defined matrix product with each other, and consequently, we obtain a symbol calculus (Corollary \ref{symbol_cal}). Also, we give criteria on the tempered weight for the corresponding magnetic pseudo-differential operators to be bounded, compact, and in the Schatten-classes 
(Theorems \ref{thm:bounded}, \ref{thm:compact}, and \ref{thm:schatten}).

The contents are organized as follows: After this introduction, we introduce the main definitions concisely. Then in Section \ref{sec:matrix_rep} we construct a tight Gabor frame and characterize our matrix representation, which is then followed by some consequences to e.g. the algebraic properties of symbols in Section \ref{sec:consequences}. Section \ref{sec:schatten} contains results on boundedness, compactness, and Schatten-class properties. Lastly, we discuss possible consequences and limitations of our approach including applications for the magnetic pseudo-differential super operators studied in \cite{LeeLein2022,LeeLein2025}.

\subsection{Symbol Classes and Magnetic Pseudo-differential Operators}

We fix $d\in\N$ and a regular magnetic field $B$ \cite{Mantoiu2004,CorneanHelfferPurice2024}  throughout. Such a regular magnetic field $B$ can be represented by smooth closed 2-forms on $\R^d$ with components in $BC^\infty(\R^d)$. For such a 2-form one can construct a vector potential $A$ satisfying $dA=B$ and having polynomial growth: 
$$A_k(x)=\sum_{j=1}^d\int_0^1ds\,sx_jB_{jk}(sx)$$
This is also called the magnetic potential in the transversal gauge. For this vector potential we define $\varphi(x,y):=\int_{[y,x]}A$ where $[y,x]$ is the line segment between $y$ and $x$.

\noindent\textbf{Tempered weight:} \cite{Nicola2010} We call $M\colon\R^{2d}\rightarrow(0,\infty)$ a tempered weight if there exists $a,C>0$ such that\footnote{Recall the Japanese bracket $\langle x\rangle:=\sqrt{1+\Vert x\Vert^2}$.}
\begin{equation}\label{eq:peetre}
    M(x+y,\xi+\zeta)\leq CM(x,\xi)\langle(y,\zeta)\rangle^a\quad\forall x,y,\xi,\zeta\in\R^d.
\end{equation}
It follows that
$$C^{-1}M(0,0)\langle(x,\xi)\rangle^{-a}\leq M(x,\xi)\leq CM(0,0)\langle(x,\xi)\rangle^a\quad\forall x,\xi\in\R^d$$
and for $p\in\R$ (generalized Peetre's inequality\footnote{Peetre's inequality refers to $\langle x\rangle^s\langle y\rangle^{-s}\leq C\langle x-y\rangle^{|s|}$ for all $s\in\R$ and some $C>0$ depending on $s$ \cite[Equation (0.1.2)]{Nicola2010}.})
\begin{equation}\label{eq:gen_peetre}
    M(x,\xi)^pM(y,\zeta)^{-p}\leq C^{|p|}\langle(x-y,\xi-\zeta)\rangle^{a|p|}\quad\forall x,y,\xi,\zeta\in\R^d.
\end{equation}
Note nothing prohibits weights with decay, e.g. $(x,\xi)\mapsto\langle\xi\rangle^{-1}$ is a tempered weight. Tempered weights are closed under pointwise multiplication and taking real exponents.

\noindent\textbf{Symbol classes:} For such a tempered weight we define the symbol class
$$S_0(M):=\left\{\Phi\in C^\infty(\R^{2d})|\sup_{x,\xi\in\R^d}M(x,\xi)^{-1}|\partial^\gamma\Phi(x,\xi)|<\infty,\forall\gamma\in\N_0^{2d}\right\}.$$
This is a Fréchet space with semi-norms given by:
$$\Vert\Phi\Vert_{S_0(M),n}=\sum_{\gamma\in\N_0^{2d},|\gamma|\leq n}\sup_{x,\xi\in\R^d}M(x,\xi)^{-1}|\partial^\gamma\Phi(x,\xi)|$$
for $n\in\N_0$. 

For any two tempered weights $M_1,M_2$, if there exists $C>0$ such that $M_1\leq CM_2$ pointwise, then $S_0(M_1)$ is continuously embedded into $S_0(M_2)$.

\noindent\textbf{Magnetic pseudo-differential operator:} \cite{Mantoiu2004,Treves2006} For $t\in[0,1]$ we define the following transformation:
$$\mathcal{W}_t^A\Phi(x,y)=(2\pi)^{-d}\int d\xi e^{i\xi\cdot(x-y)}e^{i\varphi(x,y)}\Phi(tx+(1-t)y,\xi),$$
firstly for Schwartz functions $\Phi\in\mathscr{S}(\R^{2d})$ and then extended to $\mathscr{S}'(\R^{2d})$. The operator $\mathcal{W}_t^A$ is a linear homeomorphism on both spaces with the canonical topologies (this means uniform convergence on bounded sets for $\mathscr{S}'(\R^{2d})$).

Now for every tempered distribution $\Phi$ we define the $t$-quantization $\mathfrak{Op}_t^A(\Phi)$ as the unique operator with distributional kernel $\mathcal{W}_t^A\Phi$. For $\Phi\in\bigcup_M S_0(M)$, the operator $\mathfrak{Op}_t^A(\Phi)$ satisfies
\begin{equation}\label{eq:defn_magdiff}
    \langle\mathfrak{Op}_t^A(\Phi)f,g\rangle_{\mathscr{S}',\mathscr{S}}=(2\pi)^{-d}\int d\xi\int dxdye^{i\xi\cdot(x-y)}e^{i\varphi(x,y)}\Phi(tx+(1-t)y,\xi)f(y)g(x)
\end{equation}
for $f,g\in\mathscr{S}(\R^d)$.

Note $t=1/2$ is the Weyl quantization and $t=1$ is the Kohn-Nirenberg or "standard" quantization.

\section{Tight Gabor Frame and Matrix Representation}\label{sec:matrix_rep}

\subsection{Tight Gabor Frame}\label{sec:tight_gabor_frame}

We define the family of functions $(\mathcal{G}_{\tilde{\alpha}}^A)_{\tilde{\alpha}\in\Z^{2d}}$ with $\tilde{\alpha}=(\alpha,\alpha')\in\Z^d\times\Z^d$ in the following way \cite{CorneanHelfferPurice2018,CorneanHelfferPurice2024}: Let $\mathfrak{g}\colon\R^d\rightarrow\R$ be an element of $C_0^\infty(\R^d)$ such that 
$$\supp(\mathfrak{g})\subseteq(-1,1)^d,\quad\sum_{\alpha\in\Z^d}(\tau_\alpha\mathfrak{g})^2\equiv1,$$
where $\tau_yf(x)=f(x-y)$. Then for $\tilde{\alpha}\in\Z^{2d}$ we define
$$\mathcal{G}_{\tilde{\alpha}}^A\colon\R^d\ni x\mapsto (2\pi)^{-\frac{d}{2}}e^{i\varphi(x,\alpha)}\mathfrak{g}(x-\alpha)e^{i\alpha'\cdot (x-\alpha)}.$$

\begin{lemma}\label{convergence_l2}\cite[Proposition 2.2]{CorneanHelfferPurice2024}
    The family $(\mathcal{G}_{\tilde{\alpha}}^A)_{\tilde{\alpha}\in\Z^{2d}}$ defines a Parseval frame\footnote{See \cite[Definition 5.1.2]{Christensen2016}.} in $L^2(\R^d)$ and consequently\footnote{We define $\langle\cdot,\cdot\rangle_{L^2}$ to be antilinear in the first entry and linear in the second.}
    $$f=\sum_{\tilde{\alpha}\in\Z^{2d}}\langle\mathcal{G}_{\tilde{\alpha}}^A,f\rangle_{L^2}\mathcal{G}_{\tilde{\alpha}}^A$$
    holds for every $f\in L^2(\R^d)$ with unconditional convergence.
\end{lemma}

\begin{proof}
    Using at first dominated convergence, then that $\overline{e^{i\varphi(\cdot,\alpha)}}=e^{-i\varphi(\cdot,\alpha)}$, and lastly the Parseval identity for Fourier series, we have:
    \begin{equation}\label{eq:parseval}
        \begin{aligned}
        \langle g,f\rangle_{L^2}&=\sum_{\alpha\in\Z^d}\langle(\tau_\alpha\mathfrak{g})g,(\tau_\alpha\mathfrak{g})f\rangle_{L^2}=\sum_{\alpha\in\Z^d}\langle e^{i\varphi(\cdot,\alpha)}(\tau_\alpha\mathfrak{g})g,e^{i\varphi(\cdot,\alpha)}(\tau_\alpha\mathfrak{g})f\rangle_{L^2}\\
        &=\sum_{\alpha,\alpha'\in\Z^d}\langle g,\mathcal{G}_{\tilde{\alpha}}^A\rangle_{L^2}\langle \mathcal{G}_{\tilde{\alpha}}^A,f\rangle_{L^2}
        \end{aligned}
    \end{equation}
    Specifically:
    \begin{align*}
        \Vert f\Vert_{L^2}^2=\sum_{\alpha,\alpha'\in\Z^d}|\langle \mathcal{G}_{\tilde{\alpha}}^A,f\rangle_{L^2}|^2,
    \end{align*}
    which shows that $(\mathcal{G}_{\tilde{\alpha}}^A)_{\tilde{\alpha}\in\Z^{2d}}$ defines a Parseval frame.

    For the unconditional convergence of the frame decomposition, let $(\tilde{\alpha}_k)_{k\in\N}$ be some enumeration of $\Z^{2d}$. Then for $n<m$
    \begin{align*}
        \left\Vert\sum_{k\leq n}\langle\mathcal{G}_{\tilde{\alpha}_k}^A,f\rangle_{L^2}\mathcal{G}_{\tilde{\alpha}_k}^A-\sum_{k\leq m}\langle\mathcal{G}_{\tilde{\alpha}_k}^A,f\rangle_{L^2}\mathcal{G}_{\tilde{\alpha}_k}^A\right\Vert_{L^2}&\leq\sup_{\Vert g\Vert_{L^2}=1}\left|\left\langle g,\sum_{n<k\leq m}\langle\mathcal{G}_{\tilde{\alpha}_k}^A,f\rangle_{L^2}\mathcal{G}_{\tilde{\alpha}_k}^A\right\rangle\right|\\
        &\leq\sup_{\Vert g\Vert_{L^2}=1}\Vert g\Vert_{L^2}\left(\sum_{n<k\leq m}|\langle\mathcal{G}_{\tilde{\alpha}_k}^A,f\rangle_{L^2}|^2\right)^\frac{1}{2}\\
        &\leq\left(\sum_{n<k}|\langle\mathcal{G}_{\tilde{\alpha}_k}^A,f\rangle_{L^2}|^2\right)^\frac{1}{2}.
    \end{align*}
    This goes to zero as $n\rightarrow\infty$, so $\sum_{\tilde{\alpha}\in\Z^{2d}}\langle\mathcal{G}_{\tilde{\alpha}}^A,f\rangle_{L^2}\mathcal{G}_{\tilde{\alpha}}^A$ is a convergent sequence in $L^2(\R^d)$. The identity \eqref{eq:parseval} shows that the limit must be $f$.
\end{proof}

The frame $(\mathcal{G}_{\tilde{\alpha}}^A)_{\tilde{\alpha}\in\Z^{2d}}$ has strong convergence properties for Schwartz functions and the representation above for $L^2(\R^d)$-functions can be extended to every tempered distribution. Note in this regard that for $f\in L^2(\R^d)$ and $g\in\mathscr{S}(\R^d)$ we have $\langle g,f\rangle_{L^2}=\langle f,\overline{g}\rangle_{\mathscr{S}',\mathscr{S}}$.

\begin{lemma}\label{convergence_schwartz}
    {\color{white}=}
    \begin{enumerate}
        \item\label{convergence_schwartz_i} If $f\in\mathscr{S}(\R^d)$, then for any $n,m\in\N_0$ there exists $C>0,k\in\N_0$ such that
        $$\sup_{\tilde{\alpha}\in\Z^{2d}}\langle\alpha\rangle^n\langle\alpha'\rangle^m|\langle f,\overline{\mathcal{G}_{\tilde{\alpha}}^A}\rangle_{\mathscr{S}',\mathscr{S}}|<C\sum_{\gamma\in\N_0^d,|\gamma|\leq k}\Vert\langle\cdot\rangle^k\partial^\gamma f\Vert_{L^\infty}.$$
        \item\label{convergence_schwartz_ii} If $(N_{\tilde{\alpha}})_{\tilde{\alpha}\in\Z^{2d}}$ are complex numbers such that
        $$\sup_{\tilde{\alpha}\in\Z^{2d}}\langle\alpha\rangle^n\langle\alpha'\rangle^m|N_{\tilde{\alpha}}|<\infty\quad\forall n,m\in\N_0,$$
        then $\sum_{\tilde{\alpha}\in\Z^{2d}}N_{\tilde{\alpha}}\mathcal{G}_{\tilde{\alpha}}^A$ converges absolutely in Schwartz space.
    \end{enumerate}
\end{lemma}

\begin{proof}
    Let $n,m\in\N_0$ be given. Then
    $$\langle\alpha\rangle^n\langle\alpha'\rangle^m|\langle f,\overline{\mathcal{G}_{\tilde{\alpha}}^A}\rangle_{\mathscr{S}',\mathscr{S}}|\leq C_1\langle\alpha\rangle^n\left|\int dxf(x)e^{-i\varphi(x,\alpha)}\mathfrak{g}(x-\alpha)(1-\Delta)^me^{-i\alpha'\cdot (x-\alpha)}\right|$$
    for some $C_1>0$. Using integration by parts we can move $(1-\Delta)^m$ onto the product $f(x)e^{-i\varphi(x,\alpha)}\mathfrak{g}(x-\alpha)$. By Leibniz rule we get derivatives of $f$ and $\mathfrak{g}$ along with factors from the phase. Using the triangle inequality we get a sum of terms of the following kind:
    $$C_2\langle\alpha\rangle^l\int dx|\partial^\beta f(x)||\partial^\gamma\mathfrak{g}(x-\alpha)|$$
    where $C_2>0$ and $l>n$ with the increase due to the phase factors. The above terms can be bounded by
    \begin{align*}
        \langle\alpha\rangle^l\int dx|\partial^\beta f(x)||\partial^\gamma\mathfrak{g}(x-\alpha)|&\leq C_3\int dx\langle x\rangle^l|\partial^\beta f(x)|\langle x-\alpha\rangle^l|\partial^\gamma\mathfrak{g}(x-\alpha)|\\
        &\leq C_3\Vert\langle\cdot\rangle^l\partial^\beta f\Vert_{L^\infty}\int dx\langle x\rangle^l|\partial^\gamma\mathfrak{g}(x)|,
    \end{align*}
    where we used Peetre's inequality and again $C_3>0$. This finishes the proof of \ref{convergence_schwartz_i}.

    For \ref{convergence_schwartz_ii} we show that for every $\gamma\in\N_0^d,k\in\N_0$
    \begin{equation}\label{eq:absolutesum}
        \sum_{\tilde{\alpha}\in\Z^{2d}}|N_{\tilde{\alpha}}|\Vert\langle\cdot\rangle^k\partial^\gamma\mathcal{G}_{\tilde{\alpha}}^A\Vert_{L^\infty}<\infty.
    \end{equation}
    This would imply that $\sum_{\tilde{\alpha}\in\Z^{2d}}N_{\tilde{\alpha}}\mathcal{G}_{\tilde{\alpha}}^A$ is a Cauchy sequence in Schwartz space, hence convergent. To show \eqref{eq:absolutesum} we have to bound $\Vert\langle\cdot\rangle^k\partial^\gamma\mathcal{G}_{\tilde{\alpha}}^A\Vert_{L^\infty}$ appropriately and use the condition in \ref{convergence_schwartz_ii}. By Peetre's inequality we get:
    $$\Vert\langle\cdot\rangle^k\partial^\gamma\mathcal{G}_{\tilde{\alpha}}^A\Vert_{L^\infty}\leq C_1\langle\alpha\rangle^k\Vert\langle\cdot-\alpha\rangle^k\partial^\gamma\mathcal{G}_{\tilde{\alpha}}^A\Vert_{L^\infty}$$
    with $C_1>0$. Using Leibniz rule we can bound this by a sum of terms of the kind:
    \begin{align*}
        C_2\langle\alpha\rangle^k\langle\alpha'\rangle^l\Vert\langle\cdot-\alpha\rangle^k\mathfrak{g}(\cdot-\alpha)\Vert_{L^\infty}=C_2\langle\alpha\rangle^k\langle\alpha'\rangle^l\Vert\langle\cdot\rangle^k\mathfrak{g}\Vert_{L^\infty}
    \end{align*}
    with again $C_2>0$. Thus there exists a $C_3>0$ such that
    \begin{equation}\label{eq:schwartz_est_frame}
        \Vert\langle\cdot\rangle^k\partial^\gamma\mathcal{G}_{\tilde{\alpha}}^A\Vert_{L^\infty}\leq C_3\langle\alpha\rangle^k\langle\alpha'\rangle^l\quad\forall\tilde{\alpha}\in\Z^{2d},
    \end{equation}
    which together with the condition in \ref{convergence_schwartz_ii} implies \eqref{eq:absolutesum}.
\end{proof}

\begin{lemma}\label{convergence_tempered}
    {\color{white}=}
    \begin{enumerate}
        \item\label{convergence_tempered_i} If $F\in\mathscr{S}'(\R^d)$, then for some $n,m\in\N_0$ we have
        $$\sup_{\tilde{\alpha}\in\Z^{2d}}\langle\alpha\rangle^{-n}\langle\alpha'\rangle^{-m}|\langle F,\overline{\mathcal{G}_{\tilde{\alpha}}^A}\rangle_{\mathscr{S}',\mathscr{S}}|<\infty.$$
        \item\label{convergence_tempered_ii} If $(N_{\tilde{\alpha}})_{\tilde{\alpha}\in\Z^{2d}}$ are complex numbers for which there exists $n,m\in\N_0$ such that
        $$\sup_{\tilde{\alpha}\in\Z^{2d}}\langle\alpha\rangle^{-n}\langle\alpha'\rangle^{-m}|N_{\tilde{\alpha}}|<\infty,$$
        then $\sum_{\tilde{\alpha}\in\Z^{2d}}N_{\tilde{\alpha}}\mathcal{G}_{\tilde{\alpha}}^A$ converges to a tempered distribution uniformly on all bounded subsets of $\mathscr{S}(\R^d)$.
    \end{enumerate}
\end{lemma}

\begin{proof}
    The first point \ref{convergence_tempered_i} is rather straightforward: There exists $C>0$ and $k\in\N_0$ depending on $F$ such that
    $$|\langle F,f\rangle_{\mathscr{S}',\mathscr{S}}|\leq C\sum_{\gamma\in\N_0^d,|\gamma|\leq k}\Vert\langle\cdot\rangle^k\partial^\gamma f\Vert_{L^\infty}$$
    holds for all $f\in\mathscr{S}(\R^d)$. Letting $f=\mathcal{G}_{\tilde{\alpha}}^A$ and using estimates such as \eqref{eq:schwartz_est_frame}, we get
    $$\sup_{\tilde{\alpha}\in\Z^{2d}}\langle\alpha\rangle^{-n}\langle\alpha'\rangle^{-m}|\langle F,\overline{\mathcal{G}_{\tilde{\alpha}}^A}\rangle_{\mathscr{S}',\mathscr{S}}|<\infty$$
    for some $n,m\in\N_0$.

    Next, from Lemma \ref{convergence_schwartz} we know that the sum
    $$\sum_{\tilde{\alpha}\in\Z^{2d}}N_{\tilde{\alpha}}\langle\mathcal{G}_{\tilde{\alpha}}^A,f\rangle_{\mathscr{S}',\mathscr{S}}$$
    makes sense for any Schwartz function $f$ and that
    $$\sum_{\tilde{\alpha}\in\Z^{2d}}|N_{\tilde{\alpha}}||\langle\mathcal{G}_{\tilde{\alpha}}^A,f\rangle_{\mathscr{S}',\mathscr{S}}|\leq C\sum_{\gamma\in\N_0^d,|\gamma|\leq k}\Vert\langle\cdot\rangle^k\partial^\gamma f\Vert_{L^\infty}$$
    for some $C>0$ and $k\in\N_0$. Thus $\mathscr{S}(\R^d)\ni f\mapsto\sum_{\tilde{\alpha}\in\Z^{2d}}N_{\tilde{\alpha}}\langle\mathcal{G}_{\tilde{\alpha}}^A,f\rangle_{\mathscr{S}',\mathscr{S}}$ is a tempered distribution. Finally, the sums $\sum_{\tilde{\alpha}\in\Z^{2d}}N_{\tilde{\alpha}}\langle\mathcal{G}_{\tilde{\alpha}}^A,f\rangle_{\mathscr{S}',\mathscr{S}}$ converges uniformly for $f$ in a bounded subset of $\mathscr{S}(\R^d)$ since $\sum_{\gamma\in\N_0^d,|\gamma|\leq k}\Vert\langle\cdot\rangle^k\partial^\gamma f\Vert_{L^\infty}$ would be bounded.
\end{proof}

\subsection{Matrix Representation}
For every continuous linear map $T\colon\mathscr{S}(\R^d)\rightarrow \mathscr{S}'(\R^d)$ we get a matrix representation using Lemma \ref{convergence_schwartz} and Lemma \ref{convergence_tempered} (see also \cite[Proposition 2.4]{CorneanHelfferPurice2024}):
\begin{equation}\label{eq:matrix_decomposition}
    Tf=\sum_{\tilde{\alpha},\tilde{\beta}\in\Z^{2d}}\langle T\mathcal{G}_{\tilde{\beta}}^A,\overline{\mathcal{G}_{\tilde{\alpha}}^A}\rangle_{\mathscr{S}',\mathscr{S}}\langle f,\overline{\mathcal{G}_{\tilde{\beta}}^A}\rangle_{\mathscr{S}',\mathscr{S}}\mathcal{G}_{\tilde{\alpha}}^A=\sum_{\tilde{\alpha},\tilde{\beta}\in\Z^{2d}}\mathbb{M}_{\tilde{\alpha},\tilde{\beta}}^A(T)\langle f,\overline{\mathcal{G}_{\tilde{\beta}}^A}\rangle_{\mathscr{S}',\mathscr{S}}\mathcal{G}_{\tilde{\alpha}}^A
\end{equation}
where
$$\mathbb{M}_{\tilde{\alpha},\tilde{\beta}}^A(T):=\langle T\mathcal{G}_{\tilde{\beta}}^A,\overline{\mathcal{G}_{\tilde{\alpha}}^A}\rangle_{\mathscr{S}',\mathscr{S}}.$$
We let $\mathbb{M}^A(T)$ denote the matrix $\left(\mathbb{M}_{\tilde{\alpha},\tilde{\beta}}^A(T)\right)_{\tilde{\alpha},\tilde{\beta}\in\Z^{2d}}$. Note the distributional kernel of $T$ is given by:
$$\sum_{\tilde{\alpha},\tilde{\beta}\in\Z^{2d}}\mathbb{M}_{\tilde{\alpha},\tilde{\beta}}^A(T)\mathcal{G}_{\tilde{\alpha}}^A\otimes\overline{\mathcal{G}_{\tilde{\beta}}^A}$$
One then obtains the result (generalization of \cite[Theorem 3.1]{CorneanHelfferPurice2024}):

\begin{theorem}\label{thm:matrix_rep}
    Fix $t\in[0,1]$ and a tempered weight $M$. Then:
    \begin{enumerate}
        \item\label{matrix_rep_i} For any $n,m\in\N_0$ there exists $C>0,k\in\N_0$ such that for all $\Phi\in S_0(M)$:
        \begin{equation}\label{eq:matrix_specific}
        \begin{aligned}
            \sup_{\tilde{\alpha},\tilde{\beta}\in\Z^{2d}}&\langle\alpha-\beta\rangle^n\langle\alpha'-\beta'\rangle^mM(t\alpha+(1-t)\beta,(1-t)\alpha'+t\beta')^{-1}\left|\mathbb{M}_{\tilde{\alpha},\tilde{\beta}}^A\left(\mathfrak{Op}_t^A(\Phi)\right)\right|\\
            &< C\Vert\Phi\Vert_{S_0(M),k}
        \end{aligned}
        \end{equation}
        \item\label{matrix_rep_ii} If $(N_{\tilde{\alpha},\tilde{\beta}})_{\tilde{\alpha},\tilde{\beta}\in\Z^{2d}}$ are complex numbers such that
        \begin{equation}\label{eq:matrix}
            \sup_{\tilde{\alpha},\tilde{\beta}\in\Z^{2d}}\langle\alpha-\beta\rangle^n\langle\alpha'-\beta'\rangle^mM(t\alpha+(1-t)\beta,(1-t)\alpha'+t\beta')^{-1}\left|N_{\tilde{\alpha},\tilde{\beta}}\right|< \infty\quad\forall n,m\in\N_0,
        \end{equation}
        then $\sum_{\tilde{\alpha},\tilde{\beta}\in\Z^{2d}}N_{\tilde{\alpha},\tilde{\beta}}(\mathcal{W}_t^A)^{-1}\mathcal{G}_{\tilde{\alpha}}^A\otimes\overline{\mathcal{G}_{\tilde{\beta}}^A}$ converges uniformly on compacts towards a member of $S_0(M)$ and each semi-norm of the limit are bounded by a constant times \eqref{eq:matrix} for some $n,m$.
    \end{enumerate}
\end{theorem}

Before we give the proof we note that by Stokes theorem:
$$\varphi(x,y)+\varphi(y,z)+\varphi(z,x)=\int_{<x,y,z>}B\quad\forall x,y,z\in\R^d$$
where $<x,y,z>$ is the triangle with vertices $x,y,z$. The expression $\int_{<x,y,z>}B$ is differentiable in $x,y,z$ and for every $\delta\in\N_0^{3d}$ there exists $C>0$ such that
\begin{equation}\label{eq:matrix_proof0}
    \left|\partial_{x,y,z}^\delta\int_{<x,y,z>}B\right|\leq C\langle x-y\rangle\langle y-z\rangle
\end{equation}

\noindent\textit{Proof of \ref{matrix_rep_i}.} Let $n,m\in\N_0$ and $\tilde{\alpha},\tilde{\beta}\in\Z^{2d}$ be given. Recalling \eqref{eq:defn_magdiff}, changing variables to $u:= t(x-\alpha)+(1-t)(y-\beta)$, $v:= x-\alpha-(y-\beta)$, and $\zeta:=\xi-(1-t)\alpha'-t\beta'$, and using Stokes theorem, we get:
\begin{equation}\label{eq:matrix_proof1}
    \begin{aligned}
        \mathbb{M}_{\tilde{\alpha},\tilde{\beta}}^A&(\mathfrak{Op}_t^A(\Phi))
        \\
        &=(2\pi)^{-2d}e^{i((1-t)\alpha'+t\beta')\cdot(\alpha-\beta)}e^{i\varphi(\alpha,\beta)}\int d\zeta\int dudve^{i\zeta\cdot(v+\alpha-\beta)}e^{iu\cdot(\beta'-\alpha')}\\
        &\quad\quad \times e^{i(\int_{<u-tv+\beta,\beta,\alpha>}B+\int_{<u+(1-t)v+\alpha,u-tv+\beta,\alpha>}B)}\\
        &\quad\quad \times\Phi(u+t\alpha+(1-t)\beta,\zeta+(1-t)\alpha'+t\beta')\mathfrak{g}(u-tv)\mathfrak{g}(u+(1-t)v)
    \end{aligned}
\end{equation}
First, observe that due to the support of $\mathfrak{g}$, the integrand above has compact support in $u,v$, specifically support contained in $(-2,2)^{2d}$. Secondly, abusing the exponential factors through partial integration we can create factors of $\langle\zeta\rangle^{-1}$, $\langle v+\alpha-\beta\rangle^{-1}$, and $\langle\beta'-\alpha'\rangle^{-1}$. The cost of doing this is getting derivatives of $\Phi$ and $\mathfrak{g}$, which are of no concern, and derivatives of the exponential factors with the magnetic field. The latter will leave us with factors that can be bounded by a constant times:
$$\langle u\rangle^k\langle v\rangle^k\langle\alpha-\beta\rangle^k$$
for some $k>0$ (use \eqref{eq:matrix_proof0} and Peetre's inequality). The factor $\langle u\rangle^k\langle v\rangle^k$ is bounded by the support of the integrand and $\langle\alpha-\beta\rangle^k$ can be made to disappear by additional partial integration in $\zeta$ (which does not affect the factors with magnetic fields) and Peetre's inequality $\langle v+\alpha-\beta\rangle^{-1}\leq C_1\langle v\rangle\langle \alpha-\beta\rangle^{-1}$, $C_1>0$. 

All in all, the absolute value of \eqref{eq:matrix_proof1} can be bounded by a sum of terms of the kind:
\begin{equation}\label{eq:matrix_proof2}
    \begin{aligned}
        C_2&\langle\alpha-\beta\rangle^{-n}\langle\alpha'-\beta'\rangle^{-m}\int d\zeta\int_{(-2,2)^{2d}}dudv\langle\zeta\rangle^{-d-1-a}\\
        &\times|\partial_{u,\zeta}^{\delta_1}\Phi(u+t\alpha+(1-t)\beta,\zeta+(1-t)\alpha'+t\beta')||\partial_{u,v}^{\delta_2}\mathfrak{g}(u-tv)||\partial_{u,v}^{\delta_3}\mathfrak{g}(u+(1-t)v)|
    \end{aligned}
\end{equation}
with $\delta_1,\delta_2,\delta_2\in\N_0^{2d}$ and $C_2>0$. The last step in the proof is taking care of the derivatives of $\Phi$. From $\Phi\in S_0(M)$ and \eqref{eq:gen_peetre}:
\begin{align*}
    |\partial_{u,\zeta}^{\delta_1}\Phi(u+t\alpha+(1-t)\beta&,\zeta+(1-t)\alpha'+t\beta')|\\
    &\leq C_3\Vert\Phi\Vert_{S_0(M),l}\langle u\rangle^a\langle\zeta\rangle^aM(t\alpha+(1-t)\beta,(1-t)\alpha'+t\beta')
\end{align*}
for all $u,\zeta\in\R^d$ and some $C_3,a>0$, so \eqref{eq:matrix_proof2} is bounded by
\begin{equation*}
    \begin{aligned}
        C_4\Vert\Phi\Vert_{S_0(M),l}&\langle\alpha-\beta\rangle^{-n}\langle\alpha'-\beta'\rangle^{-m}M(t\alpha+(1-t)\beta,(1-t)\alpha'+t\beta')\\
        &\times \int d\zeta\int_{(-2,2)^{2d}}dudv\langle\zeta\rangle^{-d-1}|\partial_{u,v}^{\delta_2}\mathfrak{g}(u-tv)||\partial_{u,v}^{\delta_3}\mathfrak{g}(u+(1-t)v)|
    \end{aligned}
\end{equation*}
with $l\in\N_0$ and $C_4>0$. This proves \ref{matrix_rep_i}.\qed

\noindent\textit{Proof of \ref{matrix_rep_ii}.} We want to prove that $\sum_{\tilde{\alpha},\tilde{\beta}\in\Z^{2d}}N_{\tilde{\alpha},\tilde{\beta}}(\mathcal{W}_t^A)^{-1}\mathcal{G}_{\tilde{\alpha}}^A\otimes\overline{\mathcal{G}_{\tilde{\beta}}^A}$ converges uniformly on compacts towards an element of $S_0(M)$. We do this by controlling the sum, and the sum of term-wise derivatives, in the differences $\alpha-\beta,\alpha'-\beta'$ and sums $t\alpha+(1-t)\beta,(1-t)\alpha'+t\beta'$. 

At first fix $\tilde{\alpha},\tilde{\beta}\in\Z^{2d}$. We have
\begin{equation}\label{eq:matrix_proof3}
    \begin{aligned}
        (\mathcal{W}_t^A)^{-1}\mathcal{G}_{\tilde{\alpha}}^A&\otimes\overline{\mathcal{G}_{\tilde{\beta}}^A}(u,\xi)=\int dve^{-i\xi\cdot v}e^{i\varphi(u,v)}\mathcal{G}_{\tilde{\alpha}}^A\otimes\overline{\mathcal{G}_{\tilde{\beta}}^A}(u+(1-t)v,u-tv)\\
        &=(2\pi)^{-d}e^{i(\beta\cdot\beta'-\alpha\cdot\alpha')}e^{i\varphi(\beta,\alpha)}\int dve^{i((1-t)\alpha'+t\beta'-\xi)\cdot v}e^{iu\cdot(\alpha'-\beta')}\\
        &\quad\times e^{i(\int_{<u-tv,\beta,\alpha>}B+\int_{<u-tv,u+(1-t)v,\alpha>}B)}\mathfrak{g}(u+(1-t)v-\alpha)\mathfrak{g}(u-tv-\beta)\\
        &=(2\pi)^{-d}e^{i(\beta\cdot\beta'-\alpha\cdot\alpha')}e^{i\varphi(\beta,\alpha)}\int dwe^{i((1-t)\alpha'+t\beta'-\xi)\cdot (w+\alpha-\beta)}e^{iu\cdot(\alpha'-\beta')}\\
        &\quad\times e^{i(\int_{<u-t(w+\alpha-\beta),\beta,\alpha>}B+\int_{<u-t(w+\alpha-\beta),u+(1-t)(w+\alpha-\beta),\alpha>}B)}\\
        &\quad\times\mathfrak{g}(u+(1-t)w-t\alpha-(1-t)\beta)\mathfrak{g}(u-tw-t\alpha-(1-t)\beta)
    \end{aligned}
\end{equation}
where we used the translation $w:= v-\alpha+\beta$. If we again study the factors of $\mathfrak{g}$ we see that the integrand has support contained in $(-2,2)^d$, and if $u-t\alpha-(1-t)\beta\notin(-1,1)^d$, then $(\mathcal{W}_t^A)^{-1}\mathcal{G}_{\tilde{\alpha}}^A\otimes\overline{\mathcal{G}_{\tilde{\beta}}^A}(u,\xi)=0$. We can also use partial integration and the exponential factors to generate factors of $\langle\xi-(1-t)\alpha'-t\beta'\rangle^{-1}$ at the cost of derivatives of $\mathfrak{g}$ and of the exponential factors with the magnetic field. The latter can be bounded by a constant times:
$$\langle u-tw-t\alpha-(1-t)\beta\rangle^{k_1}\langle\alpha-\beta\rangle^{k_1}\langle w\rangle^{k_1}\langle u+(1-t)w-t\alpha-(1-t)\beta\rangle^{k_1}$$
for some $k_1>0$ (\eqref{eq:matrix_proof0} and Peetre's inequality). The only factor that is not a priori bounded is $\langle\alpha-\beta\rangle^{k_1}$.

Differentiating \eqref{eq:matrix_proof3} in $u$ gives derivatives of $\mathfrak{g}$ and of the factors with the magnetic field discussed before, as well as factors of $\langle \alpha'-\beta'\rangle^{k_2}$ for some $k_2>0$. Differentiating with respect to $\xi$ only gives us factors bounded by a constant times $\langle w\rangle^{k_3}\langle\alpha-\beta\rangle^{k_3}$ for some $k_3>0$.

This all leads to $|\partial_{u,\xi}^{\delta}(\mathcal{W}_t^A)^{-1}\mathcal{G}_{\tilde{\alpha}}^A\otimes\overline{\mathcal{G}_{\tilde{\beta}}^A}(u,\xi)|$, $\delta\in\N_0^{2d}$, being bounded by a finite sum with terms of the following kind:
\begin{equation*}
    \begin{aligned}
        C_1&1_{t\alpha+(1-t)\beta+(-1,1)^d}(u)\langle\xi-(1-t)\alpha'-t\beta'\rangle^{-k_4}\langle\alpha-\beta\rangle^{k_5}\langle\alpha'-\beta'\rangle^{k_6}\\
        &\times\int_{(-2,2)^d}dv|\partial_{u,w}^{\delta_1}\mathfrak{g}(u+(1-t)w-t\alpha-(1-t)\beta)||\partial_{u,w}^{\delta_2}\mathfrak{g}(u-tw-t\alpha-(1-t)\beta)|\\
        &\leq C_21_{t\alpha+(1-t)\beta+(-1,1)^d}(u)\langle\xi-(1-t)\alpha'-t\beta'\rangle^{-k_4}\langle\alpha-\beta\rangle^{k_5}\langle\alpha'-\beta'\rangle^{k_6}
    \end{aligned}
\end{equation*}
for some $C_1,C_2,k_4,k_5,k_6>0$ and $\delta_1,\delta_2\in\N_0^{2d}$. We now have to combine this with the properties of $(N_{\tilde{\alpha},\tilde{\beta}})_{\tilde{\alpha},\tilde{\beta}\in\Z^{2d}}$. In view of \eqref{eq:matrix} and the above we obtain the estimate:
\begin{equation*}
    \begin{aligned}
        |N_{\tilde{\alpha},\tilde{\beta}}|&|\partial_{u,\xi}^{\delta}(\mathcal{W}_t^A)^{-1}\mathcal{G}_{\tilde{\alpha}}^A\otimes\overline{\mathcal{G}_{\tilde{\beta}}^A}(u,\xi)|\\
        &\leq C_31_{t\alpha+(1-t)\beta+(-1,1)^d}(u)\langle\xi-(1-t)\alpha'-t\beta'\rangle^{-d-1-a}\\
        &\quad\times\langle\alpha-\beta\rangle^{-d-1}\langle\alpha'-\beta'\rangle^{-d-1} M(t\alpha+(1-t)\beta,(1-t)\alpha'+t\beta')^{-1}\\
        &\leq C_4M(u,\xi) 1_{t\alpha+(1-t)\beta+(-1,1)^d}(u)\langle\xi-(1-t)\alpha'-t\beta'\rangle^{-d-1}\\
        &\quad\times\langle\alpha-\beta\rangle^{-d-1}\langle\alpha'-\beta'\rangle^{-d-1}
    \end{aligned}
\end{equation*}
with $C_3,C_4>0$ using \eqref{eq:gen_peetre}. This last estimate shows that $\sum_{\tilde{\alpha},\tilde{\beta}\in\Z^{2d}}N_{\tilde{\alpha},\tilde{\beta}}(\mathcal{W}_t^A)^{-1}\mathcal{G}_{\tilde{\alpha}}^A\otimes\overline{\mathcal{G}_{\tilde{\beta}}^A}$ converges uniformly on compacts to a $C^\infty(\R^{2d})$ function, and furthermore, this limit is in $S_0(M)$.\qed

\section{Change and Composition of Symbols}\label{sec:consequences}

There are several important, yet simple, consequences of Theorem \ref{thm:matrix_rep}, which we present in this section.

First of all, if $T\colon\mathscr{S}(\R^d)\rightarrow \mathscr{S}'(\R^d)$ is continuous and linear, and its matrix $\mathbb{M}(T)$ has elements satisfying \eqref{eq:matrix}, then $T$ is a magnetic pseudo-differential operator with parameter $t$ and symbol in $S_0(M)$. Simply put:
$$\Phi=\sum_{\tilde{\alpha},\tilde{\beta}\in\Z^{2d}}\mathbb{M}_{\tilde{\alpha},\tilde{\beta}}^A(T)(\mathcal{W}_t^A)^{-1}\mathcal{G}_{\tilde{\alpha}}^A\otimes\overline{\mathcal{G}_{\tilde{\beta}}^A}$$
lies in $S_0(M)$, and both $\mathfrak{Op}_t^A(\Phi)$ and $T$ have
$$\mathcal{W}_t^A\Phi=\sum_{\tilde{\alpha},\tilde{\beta}\in\Z^{2d}}\mathbb{M}_{\tilde{\alpha},\tilde{\beta}}^A(T)\mathcal{W}_t^A(\mathcal{W}_t^A)^{-1}\mathcal{G}_{\tilde{\alpha}}^A\otimes\overline{\mathcal{G}_{\tilde{\beta}}^A}=\sum_{\tilde{\alpha},\tilde{\beta}\in\Z^{2d}}\mathbb{M}_{\tilde{\alpha},\tilde{\beta}}^A(T)\mathcal{G}_{\tilde{\alpha}}^A\otimes\overline{\mathcal{G}_{\tilde{\beta}}^A}$$
as a distributional kernel, whence $T=\mathfrak{Op}_t^A(\Phi)$.

Secondly, using the decomposition \eqref{eq:matrix_decomposition}, Theorem \ref{thm:matrix_rep}, and Lemma \ref{convergence_schwartz}, one can prove that $\mathfrak{Op}_t^A(\Phi)$, $t\in[0,1]$ and $\Phi\in\bigcup_MS_0(M)$, is a bounded operator in $\mathscr{S}(\R^d)$, i.e. $\mathfrak{Op}_t^A(\Phi)f\in\mathscr{S}(\R^d)$ for every Schwartz function $f$ and $\mathfrak{Op}_t^A(\Phi)$ is continuous as an operator $\mathscr{S}(\R^d)\rightarrow \mathscr{S}(\R^d)$. The bound in Theorem \ref{thm:matrix_rep} \ref{matrix_rep_i} also shows that $S_0(M)\ni\Phi\mapsto \mathfrak{Op}_t^A(\Phi)\in\mathcal{B}(\mathscr{S}(\R^d))$ is continuous:

\begin{corollary}\label{cor:schwartz_bounded}
    Fix $t\in[0,1]$ and a tempered weight $M$. Then $\mathfrak{Op}_t^A(\Phi)\mathscr{S}(\R^d)\subseteq\mathscr{S}(\R^d)$ for every $\Phi\in S_0(M)$, $\mathfrak{Op}_t^A(\Phi)$ is bounded on $\mathscr{S}(\R^d)$, and $S_0(M)\ni\Phi\mapsto \mathfrak{Op}_t^A(\Phi)\in\mathcal{B}(\mathscr{S}(\R^d))$ is continuous.
\end{corollary}

This is an expected result, which could be obtained directly from definitions.

\subsection{Continuous Change of Quantization}

Thirdly, the criterion in \eqref{eq:matrix} is independent of $t$, so if $(N_{\tilde{\alpha},\tilde{\beta}})_{\tilde{\alpha},\tilde{\beta}\in\Z^{2d}}$ satisfy the criterion for one $t\in\R$, it satisfies it for all. To see this, suppose $t,s\in[0,1]$ and $\tilde{\alpha},\tilde{\beta}\in\Z^{2d}$. Then from \eqref{eq:peetre}
\begin{align*}
    M(s\alpha&+(1-s)\beta,(1-s)\alpha'+s\beta')^{-1}\\
    &=M(t\alpha+(1-t)\beta+(s-t)(\alpha-\beta),(1-t)\alpha'+t\beta'+(s-t)(\alpha'-\beta'))^{-1}\\
    &\leq C\langle s-t\rangle^{2p} M(t\alpha+(1-t)\beta,(1-t)\alpha'+t\beta')^{-1}\langle\alpha-\beta\rangle^a\langle\alpha'-\beta'\rangle^a,
\end{align*}
which proves the assertion by the arbitrariness of $n,m\geq0$ in \eqref{eq:matrix}. This leads to the following:

\begin{corollary}
    Fix $t,s\in[0,1]$ and a tempered weight $M$. The map
    $$\mathfrak{L}_{t,s}\colon S_0(M)\ni\Phi\mapsto\sum_{\tilde{\alpha},\tilde{\beta}\in\Z^{2d}}\mathbb{M}_{\tilde{\alpha},\tilde{\beta}}^A(\mathfrak{Op}_t^A(\Phi))(\mathcal{W}_s^A)^{-1}\mathcal{G}_{\tilde{\alpha}}^A\otimes\overline{\mathcal{G}_{\tilde{\beta}}^A}\in S_0(M)$$
    defines a linear homeomorphism and
    $$\mathfrak{Op}_s^A(\mathfrak{L}_{t,s}\Phi)=\mathfrak{Op}_t^A(\Phi).$$
\end{corollary}

\begin{proof}
    The corollary follows directly from the above. Concerning the continuity of $\mathfrak{L}_{t,s}$, it follows from first applying Theorem \ref{thm:matrix_rep} \ref{matrix_rep_ii} to a semi-norm of $\mathfrak{L}_{t,s}\Phi$ and then Theorem \ref{thm:matrix_rep} \ref{matrix_rep_i}. It is also easy to see that $\mathfrak{L}_{t,s}^{-1}=\mathfrak{L}_{s,t}$.
\end{proof}

\subsection{Magnetic Moyal Algebra} 
Lastly, multiplying matrix elements as in \cite{CorneanHelfferPurice2024}, one can prove that for two tempered weights $M_1,M_2$, $t,s,r\in[0,1]$, $\Phi\in S_0(M_1)$, and $\Psi\in S_0(M_2)$, then
$$\mathfrak{Op}_t^A(\Phi)\mathfrak{Op}_s^A(\Psi)$$
is a magnetic pseudo-differential operator with parameter $r$ and symbol in $S_0(M_1M_2)$. This is a case of the magnetic Moyal product $\#_B$ \cite{Mantoiu2004,Iftimie2007} (at least when $t=s=r=1/2$).

To actually prove it we just have to combine the previous statements: The product makes sense by Corollary \ref{cor:schwartz_bounded}. Secondly,
\begin{align*}
    \mathbb{M}_{\tilde{\alpha},\tilde{\beta}}^A(\mathfrak{Op}_t^A(\Phi)\mathfrak{Op}_s^A(\Psi))&=\langle \mathcal{G}_{\tilde{\alpha}}^A,\mathfrak{Op}_t^A(\Phi)\mathfrak{Op}_s^A(\Psi)\mathcal{G}_{\tilde{\beta}}^A\rangle_{L^2}=\langle \mathfrak{Op}_t^A(\Phi)^*\mathcal{G}_{\tilde{\alpha}}^A,\mathfrak{Op}_s^A(\Psi)\mathcal{G}_{\tilde{\beta}}^A\rangle_{L^2}\\
    &=\sum_{\tilde{\gamma}\in\Z^{2d}}\langle \mathfrak{Op}_t^A(\Phi)^*\mathcal{G}_{\tilde{\alpha}}^A,\mathcal{G}_{\tilde{\gamma}}^A\rangle_{L^2}\langle\mathcal{G}_{\tilde{\gamma}}^A,\mathfrak{Op}_s^A(\Psi)\mathcal{G}_{\tilde{\beta}}^A\rangle_{L^2}\\
    &=\sum_{\tilde{\gamma}\in\Z^{2d}}\mathbb{M}_{\tilde{\alpha},\tilde{\gamma}}^A(\mathfrak{Op}_t^A(\Phi))\mathbb{M}_{\tilde{\gamma},\tilde{\beta}}^A(\mathfrak{Op}_s^A(\Psi)).
\end{align*}
Using this sum, \eqref{eq:peetre}, and \eqref{eq:matrix}, one can conclude that $\mathbb{M}^A(\mathfrak{Op}_t^A(\Phi)\mathfrak{Op}_s^A(\Psi))$ satisfies \eqref{eq:matrix} with $M:= M_1M_2$. Furthermore, the bound \eqref{eq:matrix_specific} shows that any semi-norm of $S_0(M)$ on the new symbol would be bounded by a constant times the product of a semi-norm of $\Phi$ and one of $\Psi$. This shows:

\begin{corollary}\label{symbol_cal}
    Fix $t,s,r\in[0,1]$ and two tempered weights $M_1,M_2$. The bilinear map
    \begin{align*}
        \mathfrak{B}_{t,s,r}\colon S_0(M_1)&\times S_0(M_2)\ni(\Phi,\Psi)\\
        &\mapsto \sum_{\tilde{\alpha},\tilde{\beta}\in\Z^{2d}}\mathbb{M}_{\tilde{\alpha},\tilde{\beta}}^A(\mathfrak{Op}_t^A(\Phi)\mathfrak{Op}_s^A(\Psi))(\mathcal{W}_r^A)^{-1}\mathcal{G}_{\tilde{\alpha}}^A\otimes\overline{\mathcal{G}_{\tilde{\beta}}^A}\in S_0(M_1M_2)
    \end{align*}
    is continuous and
    $$\mathfrak{Op}_r^A(\mathfrak{B}_{t,s,r}(\Phi,\Psi))=\mathfrak{Op}_t^A(\Phi)\mathfrak{Op}_s^A(\Psi).$$
\end{corollary}

\begin{remark}\label{adjoint}
    To justify the use of the adjoint in the above calculation we note that for $t\in[0,1],\Phi\in S_0(M)$, $M$ a tempered weight, then the domain of the adjoint $\mathfrak{Op}_t^A(\Phi)^*$ includes the space of Schwartz functions and for $f\in\mathscr{S}(\R^d)$ we explicitly have:
    $$\mathfrak{Op}_t^A(\Phi)^*f=\sum_{\tilde{\alpha},\tilde{\beta}\in\Z^{2d}}\overline{\mathbb{M}_{\tilde{\beta},\tilde{\alpha}}^A(\mathfrak{Op}_t^A(\Phi))}\langle \mathcal{G}_{\tilde{\beta}}^A,f\rangle_{L^2}\mathcal{G}_{\tilde{\alpha}}^A$$
    This also implies that $\mathfrak{Op}_t^A(\Phi)^*$ is the extension of a magnetic pseudo-differential operator with symbol in $S_0(M)$, and furthermore that there exists a linear homeomorphism $\mathfrak{L}_{t,s}$, $t,s\in[0,1]$, of $S_0(M)$ such that
    $$\langle\mathfrak{Op}_s^A(\mathfrak{L}_{t,s}\Phi)g,f\rangle_{L^2}=\langle g,\mathfrak{Op}_t^A(\Phi)f\rangle_{L^2}$$
    for all $f,g\in\mathscr{S}(\R^d)$.
\end{remark}

\section{Boundedness, Compactness, and Schatten-class Properties}\label{sec:schatten}

It was proven in \cite{CorneanHelfferPurice2024} (and earlier in \cite{Iftimie2007}) that a magnetic version of the Calderón-Vaillancourt theorem holds for magnetic pseudo-differential operators. Stated in the language of this paper:

\begin{theorem}\label{thm:bounded} \cite[Theorem 3.1]{Iftimie2007}\cite[Theorem 3.7]{CorneanHelfferPurice2024}
    When the tempered weight $M$ is bounded, then $\mathfrak{Op}_t^A(\Phi)$ is a bounded operator in $L^2(\R^d)$ for all $\Phi\in S_0(M)$ and $t\in[0,1]$, and the map $S_0(M)\ni\Phi\mapsto\mathfrak{Op}_t^A(\Phi)\in \mathcal{B}(L^2(\R^d))$ is continuous.
\end{theorem}

\begin{proof}
    For any $f\in\mathscr{S}(\R^d)$ we have the expansion:
    \begin{equation*}
        \begin{aligned}
            \Vert \mathfrak{Op}_t^A(\Phi)f\Vert_{L^2}^2&=\langle \mathfrak{Op}_t^A(\Phi)^*\mathfrak{Op}_t^A(\Phi)f,f\rangle_{L^2}=\sum_{\tilde{\alpha}\in\Z^{2d}}\langle \mathfrak{Op}_t^A(\Phi)f,\mathfrak{Op}_t^A(\Phi)\mathcal{G}_{\tilde{\alpha}}^A \rangle_{L^2}\langle \mathcal{G}_{\tilde{\alpha}}^A,f \rangle_{L^2}\\
            &=\sum_{\tilde{\alpha},\tilde{\gamma}\in\Z^{2d}}\langle f,\mathfrak{Op}_t^A(\Phi)^*\mathcal{G}_{\tilde{\gamma}}^A \rangle_{L^2}\langle \mathcal{G}_{\tilde{\gamma}}^A,\mathfrak{Op}_t^A(\Phi)\mathcal{G}_{\tilde{\alpha}}^A \rangle_{L^2}\langle \mathcal{G}_{\tilde{\alpha}}^A,f \rangle_{L^2}\\
            &=\sum_{\tilde{\alpha},\tilde{\beta}\in\Z^{2d}}\langle f,\mathcal{G}_{\tilde{\beta}}^A \rangle_{L^2}\langle \mathcal{G}_{\tilde{\alpha}}^A,f \rangle_{L^2}\sum_{\tilde{\gamma}\in\Z^{2d}}\overline{\mathbb{M}_{\tilde{\beta},\tilde{\gamma}}^A(\mathfrak{Op}_t^A(\Phi))}\mathbb{M}_{\tilde{\gamma},\tilde{\alpha}}^A(\mathfrak{Op}_t^A(\Phi))
        \end{aligned}
    \end{equation*}
    Using \eqref{eq:matrix_specific} and the Schur test \cite[Theorem 1.8]{Hedenmalm2000} on $\ell^2(\Z^{2d})$ we see that
    \begin{align*}
            \ell^2(\Z^{2d})&\ni(c_{\tilde{\beta}})_{\tilde{\beta}\in\Z^{2d}}\\
            &\mapsto\left(\sum_{\tilde{\beta}\in\Z^{2d}}\left(\sum_{\tilde{\gamma}\in\Z^{2d}}\overline{\mathbb{M}_{\tilde{\beta},\tilde{\gamma}}^A(\mathfrak{Op}_t^A(\Phi))}\mathbb{M}_{\tilde{\gamma},\tilde{\alpha}}^A(\mathfrak{Op}_t^A(\Phi))\right)c_{\tilde{\beta}}\right)_{\tilde{\alpha}\in\Z^{2d}}\in\ell^2(\Z^{2d})
    \end{align*}
    is a well-defined bounded operator with norm less than $C\Vert\Phi\Vert_{S_0(M),k}^2$ for some $C>0,k\in\N_0$. In combination with Lemma \ref{convergence_l2} one finds that
    $$\Vert \mathfrak{Op}_t^A(\Phi)f\Vert_{L^2}^2\leq C\Vert\Phi\Vert_{S_0(M),k}^2\Vert(\langle \mathcal{G}_{\tilde{\alpha}}^A,f \rangle_{L^2})_{\tilde{\alpha}\in\Z^{2d}}\Vert_{\ell^2}^2=C\Vert\Phi\Vert_{S_0(M),k}^2\Vert f\Vert_{L^2}^2,$$
    which shows that $\mathfrak{Op}_t^A(\Phi)$ is bounded and the continuity assertion.
\end{proof}

We prove similar criteria for compactness and certain Schatten-class properties.

\subsection{Compactness} As noted in some sources (\cite[Theorem 1.4.2]{Nicola2010} and \cite[Theorem 4.28]{Zworski2012}), when the tempered weight $M$ goes to zero towards infinity, i.e. $\sup_{\Vert(x,\xi)\Vert>R}M(x,\xi)\rightarrow0$ as $R\rightarrow\infty$, then pseudo-differential operators are compact. This is also the case in the magnetic setting (see also \cite[Theorem 1.1]{Iftimie2010a}):

\begin{theorem}\label{thm:compact}
    Fix a tempered weight $M$ which decays to zero at infinity. Then $\mathfrak{Op}_t^A(\Phi)$ is a compact operator in $L^2(\R^d)$ for every $t\in[0,1]$ and $\Phi\in S_0(M)$.
\end{theorem}

\begin{proof}
    The proof is actually quite simple: Note first that $\mathfrak{Op}_t^A(\Phi)$ is a bounded operator by the magnetic Calderón-Vaillancourt theorem \ref{thm:bounded}, and secondly that for every $f\in L^2(\R^d)$ we may write (Lemma \ref{convergence_l2} and $L^2(\R^d)$ version of \eqref{eq:matrix_decomposition}):
    \begin{align*}
        \mathfrak{Op}_t^A(\Phi)f&=\sum_{\tilde{\alpha},\tilde{\beta}\in\Z^{2d}}\mathbb{M}_{\tilde{\alpha},\tilde{\beta}}^A(\mathfrak{Op}_t^A(\Phi))\langle \mathcal{G}_{\tilde{\beta}}^A,f\rangle_{L^2}\mathcal{G}_{\tilde{\alpha}}^A\\
        &=\sum_{\tilde{\alpha}\in\Z^{2d}}\mathcal{G}_{\tilde{\alpha}}^A\left(\sum_{\tilde{\beta}\in\Z^{2d}}\mathbb{M}_{\tilde{\alpha},\tilde{\beta}}^A(\mathfrak{Op}_t^A(\Phi))\langle \mathcal{G}_{\tilde{\beta}}^A,f\rangle_{L^2}\right)
    \end{align*}
    So $\mathfrak{Op}_t^A(\Phi)$ can be written as the composition of three maps:
    \begin{equation}\label{eq:three_map_decomposition}
        \begin{aligned}
            L^2(\R^d)\ni f&\mapsto (\langle \mathcal{G}_{\tilde{\beta}}^A,f\rangle_{L^2})_{\tilde{\beta}\in\Z^{2d}}\in\ell^2(\Z^{2d})\\
            \ell^2(\Z^{2d})\ni(c_{\tilde{\beta}})_{\tilde{\beta}\in\Z^{2d}}&\mapsto\left(\sum_{\tilde{\beta}\in\Z^{2d}}\mathbb{M}_{\tilde{\alpha},\tilde{\beta}}^A(\mathfrak{Op}_t^A(\Phi))c_{\tilde{\beta}}\right)_{\tilde{\alpha}\in\Z^{2d}}\in\ell^2(\Z^{2d})\\
            \ell^2(\Z^{2d})\ni(c_{\tilde{\beta}})_{\tilde{\beta}\in\Z^{2d}}&\mapsto\sum_{\tilde{\alpha}\in\Z^{2d}}c_{\tilde{\alpha}}\mathcal{G}_{\tilde{\alpha}}^A\in L^2(\R^d)
        \end{aligned}
    \end{equation}
    The first and third are bounded \cite[Lemma 3.2.1 and Theorem 3.2.3]{Christensen2016}. We will prove that the second is compact. 
    
    From \eqref{eq:peetre} we get
    $$M(\tilde{\alpha})^{-1}\leq CM(t\alpha+(1-t)\beta,(1-t)\alpha'+t\beta')^{-1}\langle \tilde{\alpha}-\tilde{\beta}\rangle^a$$
    for some $a,C>0$. Using this with Theorem \ref{thm:matrix_rep} we have
    \begin{align*}
        \sup_{\tilde{\alpha},\tilde{\beta}\in\Z^{2d}}M(\tilde{\alpha})^{-1}\langle\alpha-\beta\rangle^n\langle\alpha'-\beta'\rangle^m\left|\mathbb{M}_{\tilde{\alpha},\tilde{\beta}}^A\left(\mathfrak{Op}_t^A(\Phi)\right)\right|<\infty
    \end{align*}
    for all $n,m\in\N_0$. Thus the second map in \eqref{eq:three_map_decomposition} can be written as a composition of the maps:
    \begin{equation*}
        \begin{aligned}
            \ell^2(\Z^{2d})\ni(c_{\tilde{\beta}})_{\tilde{\beta}\in\Z^{2d}}&\mapsto\left(\sum_{\tilde{\beta}\in\Z^{2d}}M^{-1}(\tilde{\alpha})\mathbb{M}_{\tilde{\alpha},\tilde{\beta}}^A(\mathfrak{Op}_t^A(\Phi))c_{\tilde{\beta}}\right)_{\tilde{\alpha}\in\Z^{2d}}\in\ell^2(\Z^{2d})\\
            \ell^2(\Z^{2d})\ni(c_{\tilde{\alpha}})_{\tilde{\alpha}\in\Z^{2d}}&\mapsto(M(\tilde{\alpha})c_{\tilde{\alpha}})_{\tilde{\alpha}\in\Z^{2d}}\in\ell^2(\Z^{2d})
        \end{aligned}
    \end{equation*}
    The first is bounded (use the Schur test as in the proof of Theorem \ref{thm:bounded}) and the second is compact since it is multiplication by a sequence decaying to zero at infinity.
\end{proof}

\subsection{Schatten-class Properties} As a last point of the paper we want to establish a connection between $M$ and the $p$-Schatten-class properties of the operator $\mathfrak{Op}_t^A(\Phi)$. This has already been done to some extent, see \cite[Theorem 1.2 and Corollary 1.5]{Athmouni2018}. 

The paper \cite{Zhu2015} establishes an interesting connection between Schatten-classes and frames, and we want to utilize that for our matrix representation. Specifically, \cite[Theorem B]{Zhu2015} states that for a bounded operator $T$ on some Hilbert space $H$, $T$ being of Schatten class $\mathcal{B}_p(H)$ with parameter $p$, $p\in(0,2]$, is equivalent with
$$\sum_n\Vert Tf_n\Vert_H^p<\infty$$
for some frame $(f_n)_n$ in $H$. Also 
$$\Vert T\Vert_{\mathcal{B}_p}=\inf\left(\sum_n\Vert Tf_n\Vert_H^p\right)^\frac{1}{p}$$
where the infimum is taken over all Parseval frames.

This leads to:

\begin{theorem}\label{thm:schatten}
    Fix a tempered weight $M$ which is also a member of $L^p(\R^{2d})$ with $p\in(0,\infty)$. Then $\mathfrak{Op}_t^A(\Phi)$ is of Schatten-class $\mathcal{B}_p(L^2(\R^d))$ for every $t\in\R$ and $\Phi\in S_0(M)$, and the map $S_0(M)\ni\Phi\rightarrow\mathfrak{Op}_t^A(\Phi)\in \mathcal{B}_p(L^2(\R^d))$ is continuous.
\end{theorem}

Before the proof we note, as a consequence of \eqref{eq:peetre}, that $M$ is in $L^p(\R^{2d})$ if and only if $M$ is measurable and the values of $M$ on a lattice is $p$-summable, e.g. $(M(\tilde{\gamma}))_{\tilde{\gamma}\in\Z^{2d}}\in\ell^p(\Z^{2d})$.

\begin{proof}
    Suppose $p\in(0,2]$. Since $M$ satisfies \eqref{eq:peetre} and is in $L^p(\R^{2d})$ we know that $M$ goes towards zero at infinity, so by the Theorem \ref{thm:compact} $\mathfrak{Op}_t^A(\Phi)$ is a compact operator in $L^2(\R^d)$. Now
    \begin{align*}
        \sum_{\tilde{\gamma}\in\Z^{2d}}\Vert \mathfrak{Op}_t^A(\Phi)\mathcal{G}_{\tilde{\gamma}}^A\Vert_{L^2}^p&=\sum_{\tilde{\gamma}\in\Z^{2d}}|\langle\mathfrak{Op}_t^A(\Phi)\mathcal{G}_{\tilde{\gamma}}^A,\mathfrak{Op}_t^A(\Phi)\mathcal{G}_{\tilde{\gamma}}^A\rangle_{L^2}|^{\frac{p}{2}}\\
        &\leq\sum_{\tilde{\gamma},\tilde{\alpha}\in\Z^{2d}}|\langle\mathfrak{Op}_t^A(\Phi)\mathcal{G}_{\tilde{\gamma}}^A,\mathcal{G}_{\tilde{\alpha}}^A\rangle_{L^2}\langle\mathcal{G}_{\tilde{\alpha}}^A,\mathfrak{Op}_t^A(\Phi)\mathcal{G}_{\tilde{\gamma}}^A\rangle_{L^2}|^{\frac{p}{2}}.
    \end{align*}
    Using Theorem \ref{thm:matrix_rep} \ref{matrix_rep_i} with $t=\frac{1}{2}$ and \eqref{eq:peetre} we get
    \begin{align*}
        \sum_{\tilde{\gamma},\tilde{\alpha}\in\Z^{2d}}|\langle\mathfrak{Op}_t^A(\Phi)\mathcal{G}_{\tilde{\gamma}}^A,\mathcal{G}_{\tilde{\alpha}}^A\rangle_{L^2}&\langle\mathcal{G}_{\tilde{\alpha}}^A,\mathfrak{Op}_t^A(\Phi)\mathcal{G}_{\tilde{\gamma}}^A\rangle_{L^2}|^{\frac{p}{2}}\\
        &\leq C\Vert\Phi\Vert_{S_0(M),n}^p\sum_{\tilde{\gamma},\tilde{\alpha}\in\Z^{2d}}\langle\alpha-\gamma\rangle^{-d-1}\langle\alpha'-\gamma'\rangle^{-d-1}M(\tilde{\gamma})^p<\infty.
    \end{align*}
    Hence \cite[Theorem B]{Zhu2015} tells us that $\mathfrak{Op}_t^A(\Phi)$ is in $\mathcal{B}_p(L^2(\R^d))$ and, in combination with the above, gives us the norm estimate.

    Now the strategy for proving the statement when $p>2$ is to reduce it to the case of $p\leq2$. Suppose $p>2$. Then we still have that $M$ goes to zero at infinity, so $\mathfrak{Op}_t^A(\Phi)$ is a compact operator. Next we shall use the following repeatedly: For a compact operator $T$ on  a Hilbert space, $T$ is of $p$-Schatten-class if and only if $T^*T$ is of $p/2$-Schatten-class, and if either condition holds:
    $$\Vert T\Vert_{\mathcal{B}_p}=\Vert T^*T\Vert_{\mathcal{B}_{\frac{p}{2}}}^\frac{1}{2}$$

    The symbol calculus of Corollary \ref{symbol_cal} and Remark \ref{adjoint} on adjoints show that $\mathfrak{Op}_t^A(\Phi)^*\mathfrak{Op}_t^A(\Phi)$ is a magnetic pseudo-differential operator with symbol in $S_0(M^2)$. Find a $k\in\N$ such that $2^k>p$. Then $(\mathfrak{Op}_t^A(\Phi)^*\mathfrak{Op}_t^A(\Phi))^{2^{k-1}}$ has a symbol in $S_0(M^{2^k})$, where $M^{2^k}$ lies in $L^{\frac{p}{2^k}}(\R^{2d})$, so by what has already been proven $(\mathfrak{Op}_t^A(\Phi)^*\mathfrak{Op}_t^A(\Phi))^{2^{k-1}}\in\mathcal{B}_{\frac{p}{2^k}}(L^2(\R^d))$. Consequently, $\mathfrak{Op}_t^A(\Phi)\in\mathcal{B}_p(L^2(\R^d))$ and
    $$\Vert\mathfrak{Op}_t^A(\Phi)\Vert_{\mathcal{B}_p}=\Vert (\mathfrak{Op}_t^A(\Phi)^*\mathfrak{Op}_t^A(\Phi))^{2^{k-1}}\Vert_{\mathcal{B}_{\frac{p}{2^k}}}^\frac{1}{2^k}\leq C\Vert\Phi\Vert_{S_0(M),n}$$
    for some $n\in\N_0$.
\end{proof}

\begin{remark}
    In \cite{CorneanHelfferPurice2024} some trace- and Hilbert-Schmidt class properties was proven for certain products of magnetic pseudo-differential operators and decaying factors. This product is itself a magnetic pseudo-differential operator with symbol in our Hörmander classes and covered by the above result.

    In contrast \cite{Athmouni2018} provides more general results when $p\geq1$ using an explicit estimate in the case $p=1$ and then interpolation.
\end{remark}

\begin{remark}
    The boundedness, compactness, and Schatten-class properties presented in Theorem \ref{thm:bounded}, \ref{thm:compact}, and \ref{thm:schatten}, can be extended to the magnetic Sobolev spaces presented in \cite{Iftimie2007} using e.g. the methods in \cite[Section 5.2]{LeeLein2025}.
\end{remark}


    

\section{Discussion}

We hope that this paper might help the study of magnetic Schrödinger operators with symbols which may have a complex decay structure, e.g. mixed growth/decay conditions in different variables. In this regard, a drawback of our article is that symbols which have both decay and growth when differentiating are not treated, see e.g. \cite{Iftimie2007,Nicola2010,Zworski2012}.

One other possible application of our results relates to the magnetic pseudo-differential super operators considered in \cite{LeeLein2022,LeeLein2025}. We conjecture that working with these more general Hörmander classes ($S_0(M)$) extended to the case of super operators could lead to interesting results for the super Moyal algebra and super semi product from \cite{LeeLein2022}, as well as extend the boundedness result in \cite{LeeLein2025} to compact and $p$-Schatten class super operators on Hilbert-Schmidt operators. Another direction would be to identify certain magnetic pseudo-differential super operators which leave the set of trace-class operators from Theorem \ref{thm:schatten} invariant, which could be useful to the study of Lindblad-type time evolutions involving magnetic pseudo-differential operators.

These and other issues are left for future work.

\backmatter

\bmhead{Acknowledgements}

This work was supported, in part, by the Danish National Research Foundation (DNRF), through the Center CLASSIQUE, grant nr. 187. We want to thank H. Cornean for help and support during the writing of this work.






\bibliography{sn-bibliography}

\end{document}